\documentclass[11pt]{article}

\usepackage{fullpage}
\usepackage{latexsym}
\usepackage{amssymb,amsfonts,amsmath,amsthm}
\usepackage{url}
\usepackage{cite}
\usepackage{color}

\def\01{\{0,1\}}

\newcommand{\ket}[1]{|#1\rangle}
\newcommand{\bra}[1]{\langle#1|}
\newcommand{\ketbra}[2]{|#1\rangle\langle#2|}
\newcommand{\norm}[1]{{\left\|{#1}\right\|}}

\newcommand{\inp}[2]{\langle{#1}|{#2}\rangle} 
\newcommand{\inpc}[2]{\langle{#1},{#2}\rangle} 
\newcommand{\Tr}{\mbox{\rm Tr}}

\newcommand{\rank}{\mbox{\rm rank}}

\newcommand{\I}{\mathbb{I}}

\newcommand{\C}{\mathbb{C}}

\newcommand{\supp}{\operatorname{supp}}
\newcommand{\diag}{\operatorname{diag}}
\newcommand{\AND}{\wedge}

\newtheorem{definition}{Definition}
\newtheorem{theorem}{Theorem}

\newtheorem{corollary}[theorem]{Corollary}

\newtheorem{proposition}{Proposition}

\title{Kochen-Specker Sets and the Rank-$1$ Quantum Chromatic Number}
\author{
Giannicola Scarpa
\thanks{CWI, Science Park 123, 1098 XG Amsterdam, the Netherlands. Supported by a Vidi grant from NWO. \mbox{Email: g.scarpa@cwi.nl}} 
\and 
Simone Severini
\thanks{Department of Computer Science, and Department of Physics \& Astronomy, University College London, WC1E 6BT London, United Kingdom. Supported by a Newton International Fellowship. Email: simoseve@gmail.com}
}

\begin{document}

\maketitle

\begin{abstract}
The quantum chromatic number of a graph $G$ is sandwiched
between its chromatic number and its clique number, which are well known NP-hard quantities. 
We restrict our attention to the rank-1 quantum chromatic number $\chi_q^{(1)}(G)$, 
which upper bounds the quantum chromatic number, but is defined under stronger constraints.
We study its relation with the chromatic number $\chi(G)$ and the
minimum dimension of orthogonal representations $\xi(G)$. 
It is known that $\xi(G) \leq \chi_q^{(1)}(G) \leq \chi(G)$. 
We answer three open questions about these relations: we give a necessary and sufficient condition to have 
$\xi(G) = \chi_q^{(1)}(G)$, we exhibit a class of graphs 
such that $\xi(G) < \chi_q^{(1)}(G)$, and we give a necessary and sufficient condition to have $\chi_q^{(1)}(G) < \chi(G)$.
Our main tools are Kochen-Specker sets, collections of vectors with a traditionally important role in the study 
of noncontextuality of physical theories, and more recently in the quantification of quantum zero-error capacities.
Finally, as a corollary of our results and a result by Avis, Hasegawa, Kikuchi, and Sasaki 
on the quantum chromatic number, we give a family of Kochen-Specker sets of growing dimension.
\end{abstract}

\section{Introduction}

The chromatic number is an important and well studied graph parameter.
The notion of  \emph{quantum} chromatic number was described in its generality
by Cameron \emph{et al. }\cite{qchrom} in 2007, but it has been studied in the
context of quantum nonlocality since the late
'90s (see the seminal work by Brassard, Cleve, and Tapp  \cite{bct:simulating}; see
also the recent survey by Galliard, Wolf and Tapp \cite{gall10} and the references therein).

The value of this notion is at least twofold: first, it has a natural
use as a tool for isolating the difference between quantum and
classical behavior, at least regarding certain aspects of entanglement-assisted 
communication; second, it sheds more light on many combinatorial parameters 
between well-known NP-hard quantities like the clique and the chromatic number 
(\emph{e.g.}, the Lov\'{a}sz $\theta $-function, the dimension of orthogonal representations, \emph{etc.}). 
The first examples that show a separation between the quantum
chromatic number and the chromatic number are based on 
constructions originally considered by Frankl and R\"{o}dl \cite{fr87} 
and Buhrman, Cleve and Widgerson \cite{bcw98}, as was pointed out by Avis \emph{et al.} \cite{Avis_QPW}.
Such a separation is particularly interesting because it is responsible for the
advantages of two-party quantum protocols which in the current literature
have been called \emph{pseudo-telepathy games} (see \cite{gall02}); moreover,
it shows a genuinely quantum behavior. 
Such considerations suggest the importance of
determining how the structure of a graph characterizes its quantum chromatic
number. For this reason, here we address the relation between 
the quantum chromatic number and other common graph-theoretic parameters. 

Specifically, we focus on the rank-1 quantum chromatic number. This quantity
is obtained by using only rank-1 measurement operators in the 
protocol which motivates the general definition (see Section \ref{sec:qchrom} below). 
Essentially, it is the minimum dimension of unitary matrices assigned to the vertices of a graph,
such that matrices of adjacent vertices satisfy a natural orthogonality
condition. There is an obvious correspondence with orthogonal representation of
graphs  as introduced by Lov\'{a}sz \cite{lovasz:shannon}
in the context of zero-error information theory. (See also \cite{lovasz_geom_repr}. A more recent equivalent
expression is vector coloring \cite{vector_coloring}.) 

In the present paper, we establish a precise
connection between the minimum dimension of orthogonal representations $\xi(G)$, the chromatic number $\chi(G)$ 
and the rank-1 quantum chromatic number $\chi_q^{(1)}(G)$, for all graphs $G$. 
It is known that $\xi(G) \leq \chi_q^{(1)}(G) \leq \chi(G)$. We answer three open questions.
First, we prove that the rank-1 quantum
chromatic number is equal to the minimum dimension of the orthogonal
representation of a particular Cartesian product graph. 
Second, we exhibit graphs where $\xi(G) < \chi_q^{(1)}(G)$, thereby solving an open problem
stated in \cite{qchrom}. The proof technique is not based on a specific example,
but on a general result which connects rank-1 quantum chromatic number and
Kochen-Specker sets. These are collections of vectors originally used to
prove the inadequacy of local hidden variable theories to model quantum
mechanical behavior deterministically \cite{Kochen-Specker,KSvectors}. 
Third, using a similar connection we exhibit graphs where $\chi_q^{(1)}(G) < \chi(G)$.
We show that there is a 
quantum advantage in playing the graph coloring game described in Section 
\ref{sec:qchrom}, whenever the strategy adopted by the two players gives
a so-called \emph{weak Kochen-Specker set}. The observation is valuable
because until now the only examples of the separation were some
orthogonality graphs, specifically the Hadamard graphs considered by Avis
\emph{et al.} \cite{Avis_QPW} and introduced in \cite{fr87}, and an isolated example
with 18 vertices \cite{qchrom}. 

Avis \emph{et al.} \cite{Avis_QPW} proved that for all Hadamard graphs on $n=4m$
vertices with $m \geq 3$ it is possible to win the graph coloring game with probability
one and a quantum advantage. Combining the result with our discussion,
we can say that all such graphs provide
proofs of the Kochen-Specker theorem.

The remainder of the paper is organized as follows. Section \ref{sec:qchrom} is
devoted to preliminary definitions. 
Section \ref{sec:relation} contains our main results,
reporting facts about the relation between rank-1 quantum
chromatic number, orthogonal rank and chromatic number. The proof
techniques are based on linear algebra. Section \ref{sec:concl} delineates some
venues for further research. We conclude the paper with a list of open
problems, for example investigating the quantum chromatic number in
the context of zero-error information theory.

\section{Quantum Chromatic Number} \label{sec:qchrom}


In this section we define the concept of quantum chromatic number of a graph. 
For the sake of completeness and to fix some useful facts, 
we present a comprehensive statement about its basic properties. 
This is done by extending and completing some observations contained in \cite{qchrom}.
We give particular attention to a more constrained notion, the so-called rank-1 quantum chromatic number. 

We assume familiarity with the basics of quantum information theory. The reader can find a good introduction 
in \cite[Chapter 2]{NielsenChuang}. 
All graphs considered in this paper are \emph{simple graphs}, 
i.e. finite, unweighted, undirected graphs containing no self-loops or multiple edges.
For a graph $G$, we denote its vertex set with $V(G)$ and its edge set with $E(G)$, unless otherwise specified. 
A \emph{proper $c$-coloring} of a graph is an assignment of $c$ colors to the vertices of the graph
such that every two adjacent vertices have different colors. 
The \emph{chromatic number} of a graph $G$, denoted by $\chi(G)$, is the minimum number of colors $c$ such that 
there exists a proper $c$-coloring of $G$.

Define the \emph{coloring game} for $G=(V,E)$ as follows.
Two players, Alice and Bob, claim that they have a proper $c$-coloring for $G$. 
A referee wants to test this claim with a one-round game, so he forbids communication between the players 
and separately asks Alice the color $\alpha$ for the vertex $v$ and Bob the color $\beta$ for the vertex $w$. 
The players win the game if the following holds:
\begin{itemize}
\item If $v=w$, then $\alpha=\beta$
\item If $(v,w)\in E$, then $\alpha \neq \beta$
\item $\alpha,\beta \in \{ 1,\dots,c \} $.
\end{itemize}

A classical strategy consists of two deterministic functions $c_A: V \rightarrow [c]$ for Alice and 
\mbox{$c_B: V \rightarrow [c]$} for Bob.
A little thought will show that to win with probability 1, we must have $c_A = c_B$ 
(to satisfy the first condition) and that $c_A$ must be a valid $c$-coloring of the graph 
(for the second and third conditions). 
Therefore, classical players cannot win the game with probability $1$ using less than $\chi(G)$ colors. 

A quantum strategy for the coloring game uses an entangled state $\ket{\psi}$ of local dimension $d$ and 
two families of POVMs: for all $v\in V$, Alice has $\{E_{v\alpha}\}_{\alpha=1,\dots,c}$ 
and Bob has $\{F_{v\beta}\}_{\beta=1,\dots,c}$.
According to her input $v$, Alice applies the corresponding POVM $\{E_{v\alpha}\}_{\alpha=1,\dots,c}$ 
to her part of the entangled state and outputs the outcome $\alpha$. Bob acts similarly and outputs $\beta$.
The requirements for the game translate into the following \emph{consistency conditions}.
Alice and Bob win the coloring game with certainty, using a quantum strategy as described above, 
if and only if
\begin{equation}\forall v\in V, \forall \alpha \neq \beta, \ \bra{\psi} E_{v\alpha} \otimes F_{v\beta} \ket{\psi} = 0 \label{eq:consistency1}\end{equation}
\begin{equation} \forall (v,w)\in E, \forall \alpha, \ \bra{\psi} E_{v\alpha} \otimes F_{w\alpha} \ket{\psi} = 0 \label{eq:consistency2}.\end{equation}
In this case, we call the strategy a \emph{winning strategy} or a \emph{quantum $c$-coloring of $G$}.
Note that we do not bound the dimension of the entangled state or the rank of the measurement operators, 
we only care about the \textit{number} of measurement operators needed to win the game with certainty. 
We are now ready to give our central definition.

\begin{definition} For all graphs $G$, the \textbf{quantum chromatic number} $\chi_q(G)$
 is the minimum number $c$ such that there exists a quantum $c$-coloring of $G$.
\end{definition}

We now show that without loss of generality a quantum $c$-coloring has a clean and simple structure, 
that we denote as quantum $c$-coloring in \emph{normal form}.
The next proposition is a collection of statements from \cite{qchrom}, expanded and rearranged,
that are useful to direct our discussion. For all complex matrices $A$, let $\overline A$ denote the complex conjugate of $A$,
i.e. the matrix obtained from $A$ by taking the complex conjugate of each entry.

\begin{proposition} \label{wlog_strategies}
If $G$ has a quantum $c$-coloring, then there exists a quantum $c$-coloring of $G$ in \emph{normal form}, with the following properties:
\begin{enumerate}
 \item All POVMs are projective measurements with $c$ projectors of rank $r$.
 \item The state $\ket{\psi}$ is the maximally entangled state of local dimension $rc$.
 \item Alice's projectors are related to Bob's as follows: for all $v,\alpha$, $E_{v\alpha} = \overline{F_{v\alpha}}$.
 \item The consistency conditions \eqref{eq:consistency1} and  \eqref{eq:consistency2} can be expressed as the single condition:
    \begin{equation} \forall (v,w) \in E, \forall \alpha \in [c], \ \inpc{E_{v\alpha}}{E_{w\alpha}} = 0. \label{eq:consistency_new}\end{equation}
\end{enumerate}

\end{proposition}
\begin{proof} 
We start with a generic winning strategy consisting of entangled state $\ket{\psi'''}$, and
POVMs $\{E'''_{v\alpha}\}_{v\in V, \alpha \in [c]}$ for Alice 
and $\{F'''_{w\beta}\}_{w \in V, \beta \in [c]}$ for Bob. 
Then, we will gradually construct an equivalent strategy with the desired properties.
We will prove the statements in a few steps. Each bullet in the following list is a small statement
that is proved right after. At the end of the steps, we will have the final strategy in normal form.
The number of prime symbols of the notation for the entangled state and the 
POVM elements will reduce as soon as we get close to the final strategy, consisting of
$\ket{\psi}$ and $\{E_{v\alpha}\}_{v\in V, \alpha \in [c]}$, $\{F_{w\beta}\}_{w \in V, \beta \in [c]}$.

\begin{itemize}
\item \emph{The entangled state has full Schmidt rank.} 

Start with the entangled state $\ket{\psi'''}$, with local dimension $d'$. 
Consider the Schmidt decomposition $\ket{\psi'''} = \sum_{i= 0}^{d'-1} \lambda_i \ket{i}\ket{i}$, 
where without loss of generality $\{\ket{i}\}_{i\in \{0,\dots,d'-1\}}$ is the computational basis. 
Say there are $d$ nonzero $\lambda_i$.
Then we define the new entangled state as $\ket{\psi''} = \sum_{i: \lambda_i \neq 0} \lambda_i \ket{i}\ket{i}$, and
restrict the measurement operators to the respective supports of the reduced states as follows.
Consider the projector $P = \sum_{i: \lambda_i \neq 0}\ketbra{i}{i}$.
Then for all the POVM elements of Alice define $E''_{v\alpha} = PE'''_{v\alpha}P$, and do the same for Bob's POVM elements.
These restricted POVMs are valid measurements on $\ket{\psi''}$, and they still form a valid quantum coloring:
$\sum_{\alpha} E''_{v\alpha} = \I $ (on the $d$-dimensional subspace on which $P$ projects) and
$$\bra{\psi''} E''_{v\alpha} \otimes F''_{v\beta} \ket{\psi''} = \bra{\psi'''} P E'''_{v\alpha} P \otimes P F'''_{v\beta} P \ket{\psi'''} 
= \bra{\psi'''} E'''_{v\alpha} \otimes F'''_{v\beta} \ket{\psi'''}. $$
We have that $\ket{\psi''}$ has full Schmidt rank $d$ and together with the new POVMs is a winning strategy.

\item \emph{All POVM elements are projectors.} \\
With some abuse of notation, we identify a projector with the support of the space on which it projects. 
We denote the support of an operator $A$ by $\supp(A)$. Define the inner product $\inpc{A}{B} \equiv \Tr(A^\dagger B)$.

It follows from consistency condition \eqref{eq:consistency1} that
$$ \forall v, \alpha, \sum_{\beta \neq \alpha} \inpc{E''_{v\alpha}}{\Tr_B(I \otimes F''_{v\beta} \ketbra{\psi''}{\psi''} )} = 0,$$
but then we must have
$$ \forall v, \alpha, \inpc{E''_{v\alpha}}{\Tr_B(I \otimes F''_{v\alpha} \ketbra{\psi''}{\psi''} )} = 1.$$
Therefore, for all $v$ and $\alpha$ we can replace $E''_{v\alpha}$ with
\mbox{$E'_{v\alpha} = \supp(\Tr_B(I \otimes F''_{v\alpha} \ketbra{\psi''}{\psi''} ))$}, 
without loss of generality. A similar replacement can be done for Bob's POVM elements.

\item\emph {We have  for all $v, \alpha, \ E'_{v\alpha} = \overline{F'_{v\alpha}}$.} \\
Call $\rho = \Tr_B(\ketbra{\psi''}{\psi''}) = \sum_i \lambda^2_i \ketbra{i}{i} = \Tr_A(\ketbra{\psi''}{\psi''})$. 
Then 
$$E'_{v\alpha} = \supp (\Tr_B(\I \otimes F'_{v\alpha} \ketbra{\psi''}{\psi''})) = 
\supp (\sqrt{\rho} \overline{F'_{v\alpha}} \sqrt{\rho}).$$

Now, since all measurements are projective measurements, $\sum_\alpha E'_{v\alpha} = \I$ and 
$$F'_{v\alpha}F'_{v\beta} = 0 \Rightarrow \sqrt{\rho} E'_{v\alpha} \rho E'_{v\beta} \sqrt{\rho} = 0 
\Rightarrow  E'_{v\alpha} \rho E'_{v\beta} = 0.$$ 
Hence, we have:
$$ \rho = \I\rho\I = \sum_\alpha E'_{v\alpha} \rho \sum_\beta E'_{v\beta} =
\sum_{\alpha,\beta} E'_{v\alpha} \rho E'_{v\beta} = \sum_\alpha E'_{v\alpha} \rho E'_{v\alpha}.$$
This last fact implies that $\rho$ commutes with all operators 
(to see this, use the fact that  $(E'_{v\alpha})^2 = E'_{v\alpha}$). 
Hence we have, using the fact that $\ket{\psi''}$ has full Schmidt rank $d$,
\begin{equation} \label{eq:E_is_F}
E'_{v\alpha} = \supp(\sqrt{\rho} \overline{F'_{v\alpha}} \sqrt\rho) = 
\supp (\overline{F'_{v\alpha}}\rho) = \overline{F'_{v\alpha}}.
\end{equation}
                                                            
\item \emph{The state $\ket{\psi''}$ can be replaced by the maximally entangled state.} \\
We have just proven in \eqref{eq:E_is_F} that the winning strategies do not depend on the values 
of the Schmidt coefficients $\{\lambda_i\}$ of $\ket{\psi''}$, as long as they are nonzero. 
Thus we can set for all $i, \lambda_i = 1/\sqrt{d}$ and define  
$\ket{\psi'} = \frac{1}{\sqrt d} \sum_{i= 0}^{d-1}  \ket{i}\ket{i}$.

\item \emph{All projectors can be of the same rank $r$ and the maximally entangled state can have local dimension $rc$.} \\
We can extend the entangled state to $\ket{\psi} = \ket{\psi'} \otimes \ket{\phi_c} $ and then define 
new projectors for Alice $ E_{v\alpha} = \sum_{i=0}^{c-1} E'_{v,\alpha+i (\operatorname{mod} c)} \otimes \ketbra{i}{i}$ 
(and similarly for Bob). 
All have rank $r=d$ and act on the new state of local dimension $rc$. 
It is easy to see that the new projectors still satisfy the consistency conditions.

\item \emph{We can express the consistency conditions \eqref{eq:consistency1} and \eqref{eq:consistency2} 
just in terms of Alice's projectors as:
$\forall (v,w) \in E$ and $\forall \alpha$, $\inpc{E_{v\alpha}}{E_{w\alpha}} = 0$.}

We have that $\ket{\psi}$ is the maximally entangled state with local dimension $d$. 
It follows that for all $v,\alpha$ and $\beta$, 
$\Tr(E_{v\alpha} \otimes F_{w\beta} \ketbra{\psi}{\psi}) = 
\frac{1}{d} \Tr(E_{v\alpha} F_{w\beta}).$
We also have that for all $v$ and $\alpha$, $E_{v\alpha} = \overline{F_{v\alpha}}$. 
Then $\frac{1}{d} \Tr(E_{v\alpha} F_{w\alpha}) = 0$ if and only if $\inpc{E_{v\alpha}}{E_{w\alpha}} = 0$, and we can write 
the consistency conditions as wanted. 

\end{itemize}

Starting from any quantum $c$-coloring of $G$, we have constructed a quantum $c$-coloring in normal form, consisting of
$\ket{\psi}, \{E_{v\alpha}\}_{v\in V, \alpha \in [c]}, \{F_{w\beta}\}_{w \in V, \beta \in [c]}.$

\end{proof}

It is natural to distinguish between different types of quantum chromatic number according to the rank of the POVM elements used in the 
strategies of Alice and Bob.

\begin{definition} 
The \textbf{rank-$r$ quantum chromatic number} $\chi_q^{(r)}(G)$ of $G$
is the minimum number of colors $c$ such that $G$ has a quantum $c$-coloring consisting of 
projectors of rank $r$ and a maximally entangled state of local dimension $rc$.
\end{definition}
We can see that $\chi_q^{(r)}(G) \leq \chi_q^{(s)}(G)$ if $r\geq s$ .
To see this, use the procedure described in Proposition \ref{wlog_strategies} for increasing the local dimension 
of the entangled state and the rank of the projectors, without changing their number.
It follows that 
\begin{equation} \label{eq:rankrchrom}
\chi_q(G) = \min_r\{\chi_q^{(r)}(G) \}. 
\end{equation}

In this paper we restrict our attention to the rank-$1$ quantum chromatic number $\chi_q^{(1)}(G)$.
It follows from \eqref{eq:rankrchrom} that the rank-$1$ quantum chromatic number is an upper bound to the 
quantum chromatic number.
In a rank-$1$ quantum $c$-coloring, we have that the maximally entangled state has local dimension $c$ 
and that the rank-$1$ projectors for each vertex $v$ can be seen as outer products 
$\ketbra{a_{v\alpha}}{a_{v\alpha}}_{\alpha\in[c]}$ of an orthonormal basis $\{\ket{a_{v\alpha}}\}_{\alpha\in[c]}$.
Then the consistency condition \eqref{eq:consistency_new} becomes 
\begin{equation} 
\forall (v,w) \in E, \forall \alpha \in [c], \ \inp{a_{v\alpha}}{a_{w\alpha}} = 0. 
\label{eq:consistency_new_rank1}
\end{equation}
As explained in \cite{qchrom}, a rank-$1$ quantum $c$-coloring of $G$ induces a 
\textbf{matrix representation} of $G$, which is a map $\Phi: V \rightarrow \C^{c\times c}$ 
such that for all $(v,w)\in E$, $ \diag(\Phi(v)^\dagger \Phi(w)) = 0$. This is obtained as follows. 
For all vertices $v \in V$ consider the unitary matrix $U_v$ mapping the computational basis 
$\{\ket{i}\}_{i\in [c]} $ to $\{\ket{a_{v\alpha}}\}_{\alpha\in [c]}$. This is a $c\times c$ matrix and 
because of condition \eqref{eq:consistency_new_rank1}, if $(v,w)$ is an edge then 
the diagonal entries of $U_v^\dagger U_w$ are zero. 

\section{Rank-$1$ Quantum Chromatic Number, Orthogonal Rank and Chromatic Number} 
\label{sec:relation}

This section contains our main results, about the relation between rank-$1$ quantum chromatic number, 
orthogonal rank and chromatic number. 

A $c$-dimensional orthogonal representation of $G=(V,E)$ is a map $\phi: V \rightarrow \C^c$ 
such that for all $(v,w)\in E$, $ \inp{\phi(v)}{\phi(w)} = 0$. 
The \emph{orthogonal rank} of a graph $G$, denoted by $\xi(G)$, is defined as the minimum $c$ such that there exist an 
orthogonal representation of $G$ in $\C^c$. 
The results in \cite[Proposition 7]{qchrom}, gives the following: 
\begin{proposition} For all graphs $G$, 
\begin{equation*} \xi(G) \leq \chi_q^{(1)}(G) \leq \chi(G). \end{equation*}
\end{proposition}

Our main results answer some questions about the relation between these three quantities, 
that were left open in \cite{qchrom}.
For all graphs $G$, we give a necessary and sufficient condition for $\xi(G) = \chi_q^{(1)}(G)$, using a relation
between the rank-$1$ quantum chromatic number and the orthogonal representation of a particular Cartesian product.
Then, using the properties of Kochen-Specker sets in two different ways, 
we first give a class of graphs for which the rank-$1$ quantum chromatic number is strictly greater 
than the orthogonal rank. Later, for all graphs $G$, we give a necessary and sufficient condition for $\chi_q^{(1)}(G) < \chi(G)$.

\subsection{Equality between rank-$1$ quantum chromatic number and orthogonal rank} \label{sec:eq_chiq1_orth}

For all pairs of graphs $G$ and $H$, define their Cartesian product $G\square H$ as follows.
The vertex set $V(G\square H) = V(G) \times V(H)$ is the Cartesian product of the vertex sets of $G$ and $H$. 
We can therefore identify each vertex in $V(G\square H)$ with a pair of vertices from the two original graphs. 
There is an edge in $E(G\square H)$ between vertices $(v,i)$ and $(w,j)$ if either $v=w \AND (i,j)\in E(H)$ or $(v,w)\in E(G) \AND i=j$.

The following proposition will help us to characterize the graphs for which there is equality 
between orthogonal rank and the rank-$1$ quantum chromatic number.  
Let $K_c$ be the complete graph on $c$ vertices.
\begin{proposition} \label{thm:chi_xi} For all graphs $G$, 
$$ \chi_q^{(1)}(G) = \min \{c : \xi(G\square K_c) = c \}.$$
\end{proposition}
\begin{proof}
We first prove that we can map any orthogonal representation
in $\C^c$ of $G' = G\square K_c$ to a matrix representation in $\C^{c\times c}$ of $G$, and vice versa.

Let $\{1,\dots,c\}$ be the vertex set of $K_c$.  
The vertex set of $G'$ is $V(G') = V(G) \times \{1,\dots,c\}$. 
There is an edge in $E(G')$ between vertices $(v,i)$ and $(w,j)$ if either $v=w \AND i\neq j$ or $(v,w)\in E(G) \AND i=j$.
Thus an orthogonal representation $\{a_{(v,i)}\}_{(v,i)\in V(G')} $ of $G'$ can be mapped to a matrix representation of $G$ as follows:
for all $v\in V(G),$ let the $i$-th column of $U_v$ be $a_{(v,i)}/\norm{a_{(v,i)}}$. 
It is easy to check that this is a valid matrix representation in $\C^{c\times c}$ for $G$. 
Similarly we can map matrix representations of $G$ to orthogonal representations of $G'$. 

We now prove the main statement. 
Let $G$ be a graph with $\chi_q^{(1)}(G) = d$, then there exists an orthogonal representation of $G\square K_d$ 
in $d$ dimensions. We also know that $\xi(G\square K_d) \geq d$, because there exist subgraphs of $G\square K_d$ isomorphic to $K_d$ and $\xi(K_d)=d$.
Hence we have $\xi(G\square K_d) = \chi_q^{(1)}(G) = d$.
Now suppose that $\min \{c : \xi(G\square K_c) = c \} = d' < d$. Then there exists an orthogonal representation of $G\square K_{d'}$ in $\C^{d'}$.
We can map such orthogonal representation to a matrix representation for $G$ in $\C^{d'\times d'}$. 
But then $\chi_q^{(1)}(G) = d'$, contradicting the assumption.
Therefore we have $ \chi_q^{(1)}(G) = d = \min \{c : \xi(G\square K_c) = c \}$.
\end{proof}

We are now able to prove the following theorem.

\begin{theorem} \label{thm:chiq_chi_xi} 
For all graphs $G$, 
\begin{equation*}
\chi_q^{(1)}(G) = \xi(G)  \Longleftrightarrow  \xi(G \square K_{\xi(G)}) = \xi(G).  
\end{equation*}
\end{theorem}
\begin{proof}

It is easy to see that for all pairs of graphs $G,H$ we have 
$\xi(G\square H) \geq \max\{\xi(G), \xi(H)\}$ because there exist subgraphs 
of $G\square H$ isomorphic to $G$ and subgraphs of $G\square H$ isomorphic to $H$.
We also have that for all $c$, $\xi(K_c) = c$.
Hence, we have that $\xi(G\square K_c) \geq \max\{ \xi(G), c \}$.
Using this and Proposition \ref{thm:chi_xi}, we observe that 
$$ \chi_q^{(1)}(G) = \min \{c : \xi(G\square K_c) = c \} \geq \xi(G),$$ 
with equality if and only if $\xi(G \square K_{\xi(G)}) = \xi(G)$. 
\end{proof}

On the basis of Proposition \ref{thm:chi_xi} and following \cite{Hogben07} we can upper bound the rank-$1$ quantum chromatic number, 
in terms of a positive-semidefinite rank.
Let $S_n$ denote the set of $n\times n$ real symmetric matrices. Then for $A \in S_n$, the graph $\mathcal{G}(A)= (V,E)$ is the graph
with vertex set $V=\{1,\dots,n\}$ and edge set $E=\{ (i,j) : A_{ij} \neq 0 \} $.
The set of positive-semidefinite matrices of the graph $G$ is $$\mathcal{S}_+(G) = \{ A\in S_n : A \succeq 0, \ \mathcal{G}(A) = G \},$$ 
and the \emph{positive-semidefinite minimum rank} of $G$ is $$\mbox{mr}_+(G) = \min \{ \rank(A): A \in \mathcal{S}_+(G) \}. $$
From \cite[Observation 1.2]{Hogben07} we have $\xi(G) \leq \mbox{mr}_+( \overline G)$, and from Proposition \ref{thm:chi_xi} we have:
\begin{equation*}
\chi_q^{(1)}(G) \leq \min \{c : \mbox{mr}_+( \overline{G\square K_c}) = c \}.
\end{equation*}
This observation may be useful for future work about the complexity of computing the quantum chromatic number (see Section \ref{sec:concl}).

\subsection{Separation between rank-$1$ quantum chromatic number and orthogonal rank using KS sets}

We know that for all graphs $G$, we have $\xi(G) \leq \chi_q^{(1)}(G)$. It is easy to see that, given a matrix representation of $G$ in $\C^{n\times n}$,
one can obtain an orthogonal representation of $G$ in $\C^n$ (take the first row of each matrix).
We now exhibit graphs with rank-$1$ quantum chromatic number \emph{strictly larger} than the orthogonal rank.
These graphs are known as finite proofs for the Kochen-Specker theorem \cite{Kochen-Specker}.

Informally, the KS theorem states the nonexistence of a map $f: \C^3 \rightarrow \{0,1\} $ such that
for all orthonormal bases $b \subseteq \C^3$, $\sum_{u\in b} f(u) = 1$.

There are also examples of finite sets for which the KS theorem is valid, called KS sets.
\begin{definition}
A Kochen-Specker set in $\C^n$ is a set $S\subseteq \C^n$ such that there is no function \mbox{$f: \C^n \rightarrow \{0,1\}$} satisfying that
for all orthonormal bases $b \subseteq S$, $\sum_{u\in b} f(u) = 1$.
\end{definition}
There are examples of small subsets of $\C^3$ that are KS sets.
The first set of 117 vectors was given by Kochen and Specker in \cite{Kochen-Specker}, 
while the current smallest set in $\C^3$ has 31 vertices \cite{peres}.
We now prove the following theorem using such sets.

\begin{theorem}
There are graphs $G$ such that $\xi(G) < \chi_q^{(1)}(G)$.
\end{theorem}
\begin{proof}

Let $S \subseteq \C^3$ be a KS set. 
We exhibit a graph $G_{S}$ such that $\xi(G_{S})=3$ but $\chi_q^{(1)}(G_{S})>3$.
The vertices of $G_{S}$ are all the vectors in $S$, and there is an edge between orthogonal vectors.
The graph $G_{S}$ has obviously an orthogonal representation of dimension $3$. We show now that it is not $3$-colorable. 
Suppose we are able to 3-color the graph, and let $f$ be a function that maps a vector in $S$ to $1$ if it has color $1$, and to zero otherwise.
Every orthonormal basis $b \subseteq S$ is a clique in $G_{S}$, so we have $\sum_{u\in b} f(u) = 1$, contradicting the assumption that $S$ is a KS set.
In \cite[Proposition 11]{qchrom} it is proven that for all graphs $G$, $\chi_q^{(1)}(G)=3$ 
if and only if $\chi(G)=3$, so we conclude that $\chi_q^{(1)}(G_{KS})>3$.
\end{proof}

Orthogonality graphs of KS sets are not the only ones that exhibit a separation between 
rank-1 quantum chromatic number and chromatic number.
We know about the existence of a small graph with rank-1 quantum chromatic number and
chromatic number equal to 4, but orthogonal rank 3, based on \cite[Proposition 11]{qchrom}. 
It is the orthogonality graph of a set of 13 vectors of dimension 3 \cite{13vertices}.
This set is not a Kochen-Specker set, as there are no KS sets smaller than 18 vectors \cite{searchKS}.

\subsection{Separation between rank-$1$ quantum chromatic number and chromatic number using weak KS sets}

There is an interesting relation between KS sets and pseudo-telepathy games.
These are coordination protocols between separated players, where quantum strategies always succeed while classical
strategies have probability of success strictly smaller than $1$.
Renner and Wolf \cite{KS-PT} prove that every KS set induces a pseudo-telepathy game, and some pseudo-telepathy games induce KS sets. 
Using weak Kochen-Specker sets as defined below, we give a necessary and sufficient condition to have $\chi_q^{(1)}(G) < \chi(G)$ for a graph $G$.

\begin{definition}
A weak Kochen-Specker set in $\C^n$ is a set $S\subseteq \C^n$ such that if there exist a function $f: \C^n \rightarrow \{0,1\} $ satisfying that
for all orthonormal bases $b \subseteq S$, $\sum_{u\in b} f(u) = 1$, then there exist orthogonal unit vectors $u,v\in S$ such that $f(u)=f(v)=1$.
\end{definition}
Note that in the previous definition the vectors $u$ and $v$ such that $f(u)=f(v)=1$ must be in two distinct orthonormal bases.
Obviously, every KS set is also a weak KS set. 

We now prove the following:

\begin{theorem} \label{thm:chiq_chi}
 For all graphs $G$, we have that  $c=\chi_q^{(1)}(G) < \chi(G)$ if and only if for all optimal rank-$1$ strategies for the quantum coloring game 
$\{\ket{a_{v\alpha}}: v\in V, \alpha \in [c]\}$ is a weak KS set.
\end{theorem}

\begin{proof}
$\Rightarrow$ Let $c = \chi_q^{(1)}(G) < \chi(G)$ and let $S = \{\ket{a_{v\alpha}}: v\in V, \alpha \in [c]\}$ be the union of the vectors of any optimal winning strategy for Alice. 
Recall that without loss of generality Bob's strategy is the complex conjugate of Alice's (see Proposition \ref{wlog_strategies}).
We now show that if  $S$ \emph{is not} a KS set, then we can properly $c$-color the graph, contradicting the assumption that $\chi(G) > c$.
$S$ \emph{is not} a KS set if there exists a function $f: \C^n \rightarrow \{0,1\} $ with the following 2 properties:
\begin{enumerate}
\item For all orthonormal bases $b \subseteq S$, $\sum_{u\in b} f(u) = 1$ 
\item For all orthogonal unit vectors $u,v \in S$ we have $f(u)=0 $ or $f(v)=0$. 
\end{enumerate}
We can use the function $f$ defined above to $c$-color the graph as follows: $$\mbox{color}(v) = \alpha \mbox{ if } f(\ket{a_{v\alpha}}) = 1.$$
This is a proper $c$-coloring because:
\begin{enumerate}
 \item The rank-$1$ quantum coloring associates each vertex to an orthonormal basis, and the first property of $f$ guarantees that exactly one vector per basis has label 1.
 \item The second property of $f$ and the consistency condition \eqref{eq:consistency_new_rank1} ensure that we never color adjacent vertices with the same color.
\end{enumerate}

$\Leftarrow$ Let $\chi_q^{(1)}(G) = c$ and assume that for all optimal rank-$1$ strategies for the quantum coloring game, 
$\bigcup_{v\in V} \{ \ket{a_{v\alpha}} \}_{\alpha \in [c]} $ is a weak KS set.
Now suppose that it is possible to classically $c$-color the graph. 
Then for each $v\in V$ with classical color $\alpha$, define the projective measurement 
\mbox{$\{\ketbra{i+\alpha}{i+\alpha}\}_{i\in[c]}$} (where the addition is modulo $c$).
It is easy to see that this is a valid rank-$1$ quantum coloring with $c$ colors, 
and the union of its vectors consists of the computational basis only. Thus it is not a weak KS set,
because you can define a function that maps $\ket{1}$ to $1$ and all other vectors to $0$. 
This contradicts the assumption.
\end{proof}

A Hadamard graph $G_N=(V,E)$ is the graph with vertex set $V=\01^N$ and 
edge set $E=\{(u,v) \in V\times V : d(u,v)=N/2 \}$, where $d$ is the Hamming distance.
Avis \emph{et al.} \cite[Theorem 3.3]{Avis_QPW} prove that for all Hadamard graphs
$G_N$ with $N=4m$ and $m\geq 3$, $\chi_q^{(1)}(G_N) < \chi(G_N)$.
They explicitly exhibit the rank-1 winning strategies for the coloring game on such graphs. 
From Theorem \ref{thm:chiq_chi}, it follows that each of those strategy must be a weak Kochen-Specker set, and
it is shown in \cite{KS-PT} that every weak KS set of size $k$ induces a KS set in the same dimension $d$ with $O(k^2 d)$ elements.
Combining these points to our Theorem \ref{thm:chiq_chi}, we have proved the following.
\begin{corollary}
Every Hadamard graph $G_N$ with $N=4m$ and $m\geq 3$ induces a Kochen-Specker set.
\end{corollary}

\section{Conclusions} \label{sec:concl}

Our results clarify the relation between rank-1 quantum
chromatic number and various combinatorial parameters. In particular, we have
given a necessary and sufficient condition for a separation between rank-1
quantum chromatic number and chromatic number, therefore answering an open
question in \cite{qchrom}. Our
finding makes use of the structure of Kochen-Specker sets; it establishes a
new link between these objects and quantum coloring games in the graph-theoretic context. 

Our work suggests a number of directions and open problems:

\begin{itemize}
\item The relation between KS sets and rank-1 quantum chromatic number is not
a direct one. We still do not know a method to construct a graph $G$ such
that $\chi_{q}^{(1)}(G)<\chi\left(  G\right)  $, given the KS set. A seemingly
natural approach is to construct $G$ by performing some operations on the
orthogonality graph of the KS set, for example, by looking at the clique
complex (where each vertex corresponds to an orthonormal basis). 

\item So far, the only cases in which the single-shot Shannon capacity of a
graph, which actually corresponds to the independence number, has been proved
to be strictly smaller than the single-shot entanglement-assisted capacity,
consist of orthogonality graphs of KS sets \cite{cldww_zeroerror,dsw_zeroerror}. 
It is plausible to conjecture that the rank-1 quantum
chromatic number of graphs obtained from such orthogonality structures is
related to the single-shot entanglement-assisted capacity. It is an open
problem to uncover the link. 

\item Graph products have a fundamental role in combinatorial optimization
\cite{imrich}. Notably, the Shannon capacity of a graph is defined as an
asymptotic parameter of a graph product (namely, the strong product). We have
shown that the rank-1 quantum chromatic number can be computed by looking at
the orthogonal rank of a Cartesian product. The problem of finding 
and characterizing families of graphs for which $\chi_{q}^{(1)}(G)$ 
has extremal behavior can be then approached by studying Cartesian products. Once
read in this light, our results also provide a motivation to apply minimum
rank methods (see \cite{Hogben_survey}) in the context of quantum nonlocality.

\item The complexity of computing or approximating the quantum chromatic number 
of a graph is an important open problem. The relation with \emph{minimum semidefinite rank}
problems highlighted in Section \ref{sec:eq_chiq1_orth} may be an interesting starting point 
for future work on this question.
\end{itemize}

\paragraph{Acknowledgments.}

The authors thank Ronald de Wolf for useful discussions and for reading preliminary versions of the paper. 
We thank Monique Laurent for an important observation leading to Theorem \ref{thm:chiq_chi_xi} and 
we thank Leslie Hogben, Jop Bri\"et and Antonios Varvitsiotis for useful discussions.
We are also very grateful to Laura Man\v{c}inska and M\={a}ris Ozols for
pointing out that there is an orthogonality graph of a set of 13
vectors of dimension 3 with rank-1 quantum chromatic number and
chromatic number equal to 4.

\bibliographystyle{alpha}

\newcommand{\etalchar}[1]{$^{#1}$}

\end{document}